%% file: santaclaus.tex
\documentclass[letterpaper,11pt]{article}
\usepackage[hmargin=1in,vmargin=1in]{geometry}
\usepackage{amsmath, amssymb}
\usepackage{graphicx}
\usepackage{algorithmic}
\usepackage{algorithm}

\title{Santa Claus Schedules Jobs on Unrelated Machines}
\date{\today}

\author{Ola Svensson  ({\tt osven@kth.se}) \\ Royal Institute of Technology - KTH\\ Stockholm, Sweden}

\input{envs}

\newcommand{\MA}{\ensuremath{\mathcal{M}}}
\newcommand{\MAB}{\ensuremath{\mathcal{M}_{{B}}}}
\newcommand{\MAS}{\ensuremath{\mathcal{M}_{{S}}}}
\newcommand{\MAM}{\ensuremath{\mathcal{M}_M}}
\newcommand{\JO}{\ensuremath{\mathcal{J}}}
\newcommand{\CompTree}{\ensuremath{\mathcal{T}}}

\newcommand{\CLP}{\ensuremath{\mbox{[C-LP]}}}

\newcommand{\conff}[1]{\ensuremath{\mathcal{C}(#1)}}

\newcommand{\val}{\ensuremath{\textnormal{Val}}}

\begin{document}
\maketitle

\begin{abstract}
  One of the classic results in scheduling theory is the
  $2$-approximation algorithm by Lenstra, Shmoys, and Tardos for the
  problem of scheduling jobs to minimize makespan on unrelated
  machines, i.e., job $j$ requires time $p_{ij}$ if processed on
  machine $i$. More than two decades after its introduction it is
  still the algorithm of choice even in the restricted model where
  processing times are of the form $p_{ij} \in \{p_j, \infty\}$.  This
  problem, also known as the restricted assignment problem, is NP-hard
  to approximate within a factor less than $1.5$ which is also the
  best known lower bound for the general version.

  Our main result is a polynomial time algorithm that estimates the
  optimal makespan of the restricted assignment problem within a
  factor $33/17 + \epsilon \approx 1.9412 + \epsilon$, where $\epsilon
  > 0$ is an arbitrarily small constant.  The result is obtained by
  upper bounding the integrality gap of a certain strong linear
  program, known as configuration LP, that was previously successfully
  used for the related Santa Claus problem.  Similar to the strongest
  analysis for that problem our proof is based on a local search
  algorithm that will eventually find a schedule of the mentioned
  approximation guarantee, but is not known to converge in polynomial
  time.
\end{abstract}

\section{Introduction}
Scheduling on unrelated machines is the model where we are given a set
$\JO$ of  jobs to be processed without interruption on a set $\MA$ of
unrelated machines, where the time a machine $i\in \MA$ needs
to process a job $j\in \JO$ is specified by a machine and job
dependent processing time $p_{ij} \geq 0$. When considering a
scheduling problem the most common and perhaps most natural objective
function is makespan minimization. This is the problem of finding
 a schedule, also called an assignment, $\sigma: \JO\mapsto \MA$
so as to minimize the time $\max_{i\in \MA} \sum_{j\in \sigma^{-1}(i)}
p_{ij}$ required to process all the jobs.  

A classic result in scheduling theory is Lenstra, Shmoys, and Tardos'
$2$-approximation algorithm for this basic problem~\cite{LS90}. 
Their approach is based on several nice structural properties of the
extreme point solutions of a natural linear program and has become a
text book example of such techniques (see, e.\,g.,~\cite{V01}).
Complementing their positive result they also proved that the problem
is NP-hard to approximate within a factor less than $1.5$ even in the
restricted case when $p_{ij} \in \{p_j, \infty\}$ (i.\,e., when job $j$
has processing time $p_j$ or $\infty$ for each machine).  This problem
is also known as the restricted assignment problem and, although it
looks easier than the general version, the algorithm of choice has
been the same $2$-approximation algorithm as for the general version.

Despite being a prominent open problem in scheduling theory, there has
been very little progress on either the upper or lower bound since the
publication of~\cite{LS90} over two decades ago.  One of the biggest
hurdles for improving the approximation guarantee has been to obtain a
good lower bound on the optimal makespan. Indeed, the considered
linear program has been useful for generalizations such as introducing
job and machine dependent costs~\cite{ST93,Singh08} but is known to
have an integrality gap of~$2-1/|\MA|$ even in the restricted case. We
note that Shchepin and Vakhania~\cite{SV05} presented a rounding
achieving this gap slightly improving upon the approximation ratio of
$2$.

In a relatively recent paper, Ebenlendr et al.~\cite{EKS08} overcame
this issue in the special case of the restricted assignment problem
where a job can be assigned to at most two machines. Their strategy
was to add more constraints to the studied linear program, which
allowed them to prove a $1.75$-approximation algorithm for this
special case that they named Graph Balancing. The name arises
naturally when interpreting the restricted assignment problem as a
hypergraph with a vertex for each machine and a hyperedge $\Gamma(j)
= \{i\in \MA:p_{ij} = p_j\}$ for each job $j\in \JO$ that is incident
to the machines it can be assigned to. As pointed out by the authors
of~\cite{EKS08} it seems difficult to extend their techniques to hold
for more general cases. In particular, it can be seen that the
considered linear program has an integrality gap of $2$ when we allow
jobs that can be assigned to $3$ machines.

In this paper we overcome this obstacle by considering a certain
strong linear program, often referred to as configuration LP.  In
particular, we obtain the first asymptotic improvement on the approximation factor of $2$.

\begin{theorem}
  \label{thm:mainintro}
  There is a polynomial time algorithm that estimates the optimal
  makespan of the restricted assignment problem within a factor of
  $33/17 + \epsilon\approx 1.9412 + \epsilon$, where $\epsilon>0$ is
  an arbitrarily small constant.
\end{theorem}
We note that our proof gives a local search algorithm to also find a schedule with performance guarantee 
 $\frac{33}{17}$ but it is not known to converge in
polynomial time.

Our techniques are based on the recent development on the related
Santa Claus problem. In the Santa Claus problem we are given the same
input as in the considered scheduling problem but instead of wanting to
minimize the maximum we wish to maximize the minimum, i.\,e., to find an
assignment $\sigma$ so as to maximize $\min_{i\in \MA} \sum_{j\in
  \sigma^{-1}(i)} p_{ij}$. The playful name now follows from associating the
machines with kids and jobs with presents. Santa Claus' problem
then becomes to distribute the presents so as to make the least happy kid as
happy as possible.

The problem was first considered under this name by Bansal and
Sviridenko~\cite{BS06}. They formulated and used the configuration
LP to obtain an $O(\log\log \log |\MA| / \log \log |\MA|)$-approximation
algorithm for the restricted Santa Claus problem, where $p_{ij} \in
\{p_j, 0\}$. They also proved several structural properties that were
later used by Feige~\cite{Feige08} to prove that the integrality gap of the
configuration LP is in fact  constant in the restricted case. The
proof is based on repeated use of Lov\'{a}sz local lemma and was only
recently turned into a polynomial time algorithm~\cite{HSS10}.

The approximation guarantee obtained by
combining~\cite{Feige08} and~\cite{HSS10} is a large constant and the techniques do
not seem applicable to the considered problem. This is because
the methods rely on structural properties that are obtained by
rounding the input and such a rounding applied to the scheduling problem
would rapidly eliminate any advantage obtained over the current
approximation ratio of $2$.
Instead, our techniques are mainly inspired by a paper of Asadpour et
al.~\cite{AFS08} 
who gave a tighter analysis of the configuration LP for the restricted
Santa Claus problem. More specifically, they proved that the
integrality gap is lower bounded by $1/4$ by designing a local search
algorithm that eventually finds a solution with the mentioned
approximation guarantee, but is not known to converge in polynomial
time.

Similar to their approach, we formulate the configuration LP and show
that its integrality gap is upper bounded by $33/17$ by designing a
local search algorithm. As the configuration LP can be solved in
polynomial time up to any desired accuracy~\cite{BS06}, this implies
Theorem~\ref{thm:mainintro}. Although we cannot prove that the local
search converges in polynomial time, our results imply that the
configuration LP gives a polynomial time computable lower bound on the
optimal makespan that is strictly better than two. We emphasize
that all the results related to hardness of approximation remain
valid even for estimating the optimal makespan.  

Before proceeding, let us mention that unlike the restricted
assignment problem, the special case of uniform machines and that of a
fixed number of machines are both significantly easier to approximate
and are known to admit polynomial time approximation
schemes~\cite{HS88,HS76,JP2001}. Also scheduling jobs on unrelated
machines to minimize weighted completion time instead of makespan has
a better approximation algorithm with performance guarantee
$1.5$~\cite{SS02}. The Santa Claus problem has also been studied under
the name Max-Min Fair Allocation and there have been several recent
results for the general version of the problem (see e.g.~\cite{AS10,BCV09,CCK09}).

Compared to~\cite{AFS08}, our analysis is more complex and relies on
the special structure of the dual of the linear program. 
To illustrate
the main techniques, we have therefore chosen to first present the
analysis for the case of only two job sizes
(Section~\ref{sec:simple}) 
 followed by the general case in Section~\ref{sec:mainalgo}.

\section{Preliminaries}
\label{sec:prelim}
As we consider the restricted case with $p_{ij} \in \{p_j, \infty\}$,
without ambiguity, we refer to $p_j$ as the size of job $j$. For a
subset $\JO' \subseteq \JO$ we let $p\left(\JO'\right) = \sum_{j\in
  \JO'} p_j$ and often write $p(j)$ for
$p(\{j\})$, which of course equals $p_j$.

We  now give the definition of the configuration LP for the
restricted assignment problem.  Its
intuition is that a solution to the scheduling problem with
makespan $T$ assigns a set of jobs, referred to as a
configuration, to each machine of total processing time at most $T$.
Formally, we say that a subset $C \subseteq \JO$ of jobs is a
configuration for a machine $i\in \MA$ if it can be assigned without
violating a given target makespan $T$, i.e., $C\subseteq \{j: i\in \Gamma(j)\}$ and $p(C)
\leq T$.  Let $\mathcal{C}(i, T)$ be the set of configurations for
machine $i\in \MA$ with respect to the target makespan $T$.
The configuration LP has a variable $x_{i,C}$ for each configuration $C$ for machine $i$ and two sets of constraints:

\begin{equation*}
\text{\CLP} \qquad
\begin{aligned}[t]
  \sum_{C \in \conff{i,T}}x_{i, C}   & \leq 1 & i \in \MA\\
  \sum_{C\ni j}\sum_{i} x_{i,C} & \geq 1 & j \in \JO \\[-2mm]
  x & \geq 0 &  
\end{aligned}
\qquad \qquad
\end{equation*}
The first set of
constraints ensures that each machine is assigned  at most one
configuration 
and the second set of constraints says  that each job should be
assigned (at least) once.

Note that if \CLP{} is feasible with respect to some target makespan
$T_0$ then it is also feasible with respect to all $T\geq T_0$. Let
$OPT_{LP}$ denote the minimum over all such values of $T$.  Since an
optimal schedule of makespan $OPT$ defines a feasible solution to
\CLP{} with $T=OPT$, we have $OPT_{LP} \leq OPT$.  To simplify
notation we will assume throughout the paper that $OPT_{LP} =1$ and
denote $\conff{i,1}$ by $\conff{i}$. This is without loss of generality
since it can be obtained by scaling processing times.

Although \CLP{} might have exponentially many variables, it can be
solved (and $OPT_{LP}$ can be found by binary search) in polynomial time up to any
desired accuracy $\epsilon >0$~\cite{BS06}.
The strategy of~\cite{BS06} is to design a polynomial time separation
oracle for the dual and then solve it using the ellipsoid
method. 
To obtain the dual, we associate a dual variable $y_i$ with $i\in \MA$
for each constraint from the first set of constraints and a dual variable $z_j$ with
$j \in \JO$ for each constraint from the second set of
constraints. Assuming that the objective  of $\CLP{}$ is to
maximize an objective function with zero coefficients then gives the dual:

\begin{equation*}
\text{Dual of \CLP} \qquad
\begin{aligned}[t]
  \min &  \sum_{i\in \MA} y_i - \sum _{j \in \JO} z_j &\\
  y_i&\geq \sum_{j\in C} z_j & i\in \MA, C\in \conff{i} \\
  y,z & \geq  0
\end{aligned}
\end{equation*}
Let us remark that, given a candidate solution $(y^*,z^*)$, the separation oracle has to find a violated
constraint if any in polynomial time and this is just $m$
knapsack problems: for each $i\in \MA$ solve the knapsack problem with
capacity $1$ and an item with weight $p_{j}$ and profit $z_j$ for each
$j\in \JO$ with $i\in \Gamma(j)$. By rounding job sizes as explained
in~\cite{BS06}, we can thus solve \CLP{} in polynomial time up to any
desired accuracy.

\section{Overview of Techniques:  Jobs of Two Sizes}
\label{sec:simple}
We give an overview of the main techniques used by considering the
simpler case when we have jobs of two sizes: \emph{small} jobs of size
$\epsilon$ and \emph{big} jobs of size $1$. Already for this case all previously considered linear programs have an integrality gap of  $2$. In contrast we show the following for~\CLP.

\begin{theorem}
  \label{thm:simple}
  If an instance of the scheduling problem only has jobs of sizes
  $\epsilon\geq 0$ and $1$ then  \CLP{} has integrality gap at most
  $5/3+\epsilon$.
\end{theorem}

Throughout this section we let $R=2/3 + \epsilon$. The proof strategy
is to design a local search algorithm that returns a solution
with makespan at most $1+R$, assuming the \CLP{} is feasible. The
algorithm starts with a partial schedule $\sigma$ with no jobs
assigned to any machine. It will then repeatedly call a procedure,
summarized in Algorithm~\ref{algo:SimpleExtSched}, that extends the schedule
by assigning a new job until all jobs are assigned.
When assigning a new job we need to ensure that $\sigma$  will still
have  a makespan of at most $1+R$.  This might require us to also update the
schedule $\sigma$ by moving already assigned jobs. For an example,
consider Figure~\ref{fig:idea} where we have a partial schedule and
wish to assign a new big job $j_{new}$. In the first step we try to
assign $j_{new}$ to $M_1$ but discover that $M_1$ has  too high load, i.e.,
the set of jobs assigned to $M_1$ have total processing time such that
assigning $j_{new}$ to $M_1$ would violate the target makespan $1+R$.
Therefore, in the second step we try to move jobs from $M_1$ to $M_2$
but $M_2$ has also too high load. Instead, we try to move $j_{new}$ to
$M_3$. As $M_3$ already has a big job assigned we need to first
reassign it. We try to reassign it to $M_4$ in the fourth step. In the
fifth step we manage to move small jobs from $M_4$ to $M_3$, which
makes it possible to also move the big job assigned to $M_3$ to $M_4$
and finally assign $j_{new}$ to $M_3$.

\begin{figure*}[hbtp]
   \centering
   \includegraphics[width=14cm]{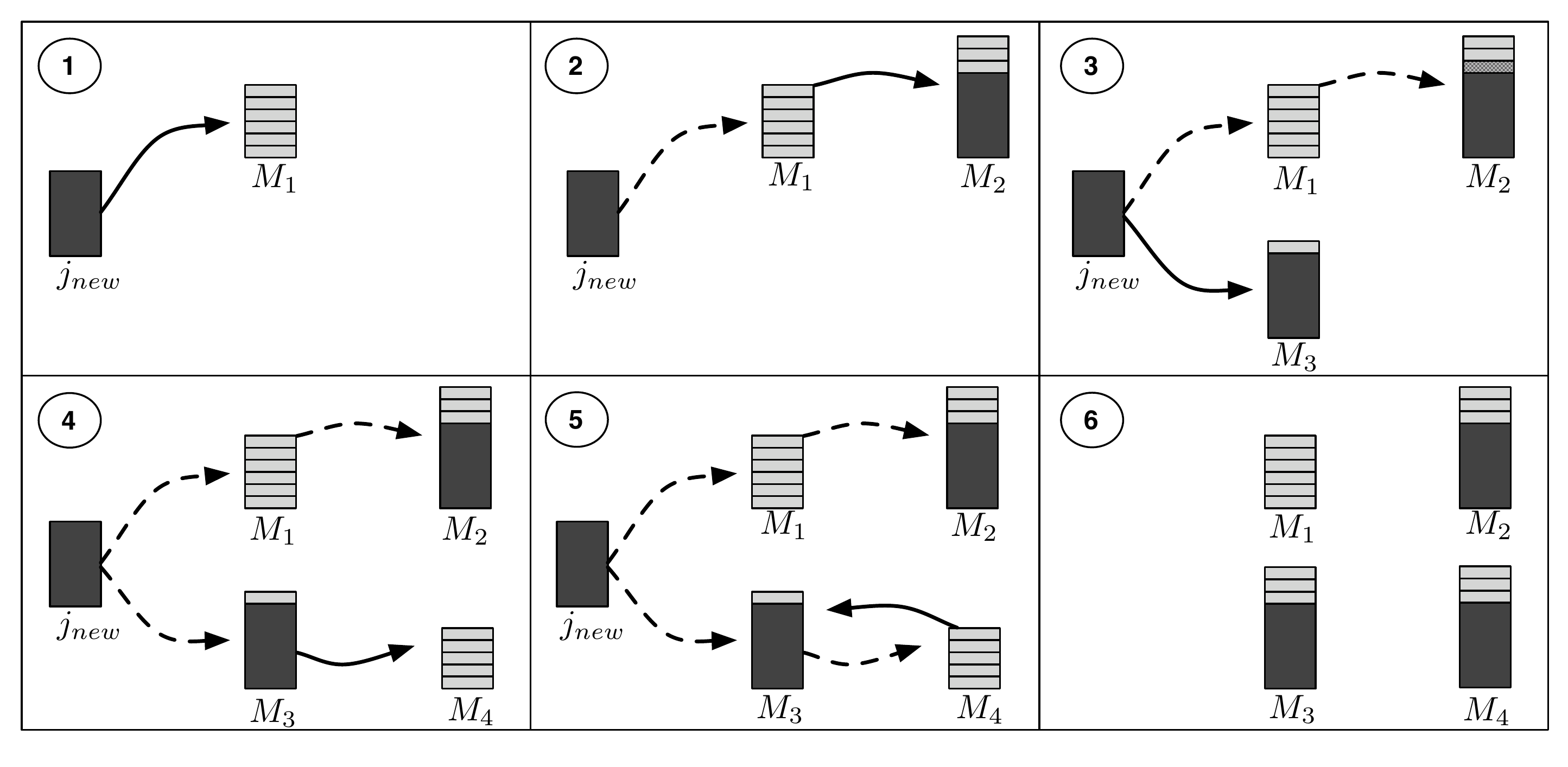}
   \caption{Possible steps when moving jobs to assign a new job $j_{new}$. Big and small jobs depicted in dark and light grey, respectively.}
\label{fig:idea}
 \end{figure*}

 \paragraph{Valid schedule and move.} As alluded to above, the
 algorithm always maintains a valid partial schedule by moving already
 assigned jobs. Let us formally define these concepts.

 \begin{definition}
   \label{def:fissched}
   A \emph{partial schedule} is an assignment $\sigma : \JO \mapsto \MA \cup
   \{TBD\}$ with the meaning that a job $j$ with $\sigma(j) = TBD$ is not assigned.
   A partial schedule is \emph{valid} if each machine $i\in \MA$ is
   assigned at
   most one big job and $p(\sigma^{-1}(i)) \leq 1+R$. 
 \end{definition}
 That $i$ is assigned at most one big job is implied here by
 $p(\sigma^{-1}(i)) \leq 1+R$ but will be used for the general case in
 Section~\ref{sec:mainalgo}.  Note also that with this notation a
 normal schedule is just a partial schedule $\sigma$ with
 $\sigma^{-1}(TBD) = \emptyset$.

 \begin{definition}
   A move is a tuple $(j,i)$ of a job $j\in \JO$ and a machine $i\in
   \Gamma_{\sigma}(j)$, where $\Gamma_{\sigma}(j) = \Gamma(j)
   \setminus\{\sigma(j)\}$ denotes the machines to which $j$ can be
   assigned apart from $\sigma(j)$.
 \end{definition}

 The main steps of the algorithm are the following.  At the start it
 will try to choose a valid assignment of $j_{new}$ to a machine,
 i.e., that can be made without violating the target makespan. If no
 such assignment exists then the algorithm adds the set of jobs that
 blocked the assignment of $j_{new}$ to the set of jobs we wish to
 move. It then repeatedly chooses a move of a job $j$ from the set of
 jobs that is blocking the assignment of $j_{new}$. If the move of $j$
 is valid then this will intuitively make more space for
 $j_{new}$. Otherwise the set of jobs blocking the move of $j$ is
 added to  the list of jobs we wish to move and the procedure will
 in the next iteration continue to move jobs recursively.

 To ensure that we will be able to eventually assign a new job
 $j_{new}$ it is important which moves we choose.  The algorithm will
 choose between certain moves that we call potential moves, defined so
 as to guarantee (i) that the procedure terminates and that (ii) if no
 potential move exists then we shall be able to prove that the dual of
 \CLP{} is unbounded, contradicting the feasibility of the primal.
 For this reason, we need to remember to which machines we have
 already tried to move jobs and which jobs we wish to move. We next
 describe how the algorithm keeps track of its history and how this
 affects which move we choose. We then describe the types of
 potential moves that the algorithm will choose between.

 \paragraph{Tree of blockers.} To remember its history,
Algorithm~\ref{algo:SimpleExtSched} has a dynamic tree $\CompTree$ of
 so-called \emph{blockers} that ``block'' moves we wish to
 do. Blockers of $\CompTree$  have both tree and linear structure. The
 linear structure is simply the order in time the blockers were added
 to $\CompTree$. To distinguish between the two we will use child and
 parent to refer to the tree structure; and after and before to refer
 to the linear structure. We also use the convention that the blockers
 $B_0, B_1, \dots, B_t$ of $\CompTree$ are indexed according to the
 linear order.
 \begin{definition}
   A blocker $B$ is a tuple that contains a subset $\JO(B)\subseteq
   \JO$ of jobs and a  machine $\MA(B)$ that takes value $\bot$ if no
   machine is assigned to the blocker.  
 \end{definition}
 To simplify notation, we refer to the machines and jobs in
 $\CompTree$ by $\MA(\CompTree)$ and $\JO(\CompTree)$,
 respectively. We will conceptually distinguish between \emph{small}
 and \emph{big} blockers and use $\MAS(\CompTree)$ and
 $\MAB(\CompTree)$ to refer to the subsets of $\MA(\CompTree)$
 containing the machines in small and big blockers, respectively. To
 be precise, this convention will add a bit to the description of a
 blocker so as to keep track of whether a blocker is small or
 big. 

 The algorithm starts by  initializing the tree $\CompTree$ with a
 special small blocker $B$ as root. Blocker $B$ is special in the
 sense that it is
 the only blocker with no machine assigned, i.e., $\MA(B) = \bot$.
 Its job set $\JO(B)$ includes the job $j_{new}$ we wish to assign.
 The next step of the procedure is to repeatedly try to move jobs,
 until we can eventually assign $j_{new}$. During its execution, the procedure
  also updates $\CompTree$ based on which move that is chosen so that 
\begin{enumerate}
\item $\MAS(\CompTree)$ contains those machines to which the algorithm will not try to move any jobs;

\item $\MAB(\CompTree)$ contains those
  machines to which the algorithm will not try to move any big jobs;

\item $\JO(\CompTree)$
  contains those jobs that the algorithm wishes to move.
\end{enumerate}
\paragraph{Potential moves.}
For a move $(j,i)$ to be useful it should be of some job $j\in \JO(
\CompTree)$ as this set contains those jobs we wish to move to make
space for the unassigned job $j_{new}$. In addition, the move $(j,i)$
should have a potential of succeeding and be to a machine $i$ where $j$
is allowed to be moved according to $\CompTree$. We refer to such
moves as \emph{potential} moves and a subset of them as
\emph{valid} moves. The difference is that for a potential move to succeed it might be necessary to recursively move other jobs whereas a valid move can be done immediately.
 With this intuition, let us now define these concepts
 formally.

\begin{definition}
  A move $(j,i)$ of a job $j\in \JO(\CompTree)$ is  a   potential
\begin{description}\itemsep0mm
\item[\textnormal{small move}:] 
 if $j$ is small and $i \not \in \MAS(\CompTree)$;
\item [\textnormal{big-to-small move:}] if $j$ is big, $i \not \in \MA(\CompTree), p(S_i) \leq R,$ and  no big job is assigned to $i$;
\item [\textnormal{big-to-big move:} ] if $j$ is big, $i \not \in \MA(\CompTree), p(S_i) \leq R,$ and  a big job is assigned to $i$;
\end{description}
where  $S_i=\{j\in \sigma^{-1}(i): j \mbox{ is small with $\Gamma_\sigma(j) \subseteq \MAS(\CompTree)$}\}$.
A potential move $(j,i)$ is \emph{valid} if the update $\sigma(j)
\leftarrow i$ results in a valid schedule.
\end{definition}
Note that $S_i$ refers to those small jobs assigned to $i$ with no
potential moves with respect to the current tree. The condition
$p(S_i) \leq R$ for big moves enforces that we do not try to move big
jobs to machines where the load cannot decrease to at most $R$ without
removing a blocker already present in $\CompTree$.
The algorithm's behavior   depends on the type of the chosen
potential move, say $(j,i)$ of a job $j \in \JO(B)$ for some blocker
$B$:
\begin{itemize}
\item If $(j,i)$ is a valid move then the schedule is updated by
$\sigma(j) \leftarrow i$. Moreover, $\CompTree$ is updated by removing
$B$ and all blockers added after $B$. This will allow us to prove that
the procedure terminates with the intuition being that $B$ blocked some move
$(j',i')$ that is more likely to succeed now after $j$ was reassigned.
\item If $(j,i)$ is a potential small or big-to-small move that is not
  valid then the algorithm adds a small blocker $B_S$ as a child to
  $B$ that consists of the machine $i$ and contains all jobs  assigned
  to $i$  that are not already in $\CompTree$. Note that after this,
  since $B_S$ is a small blocker no other jobs will be tried to be
  moved to $i$. The intuition of this being that assigning more jobs to $i$ would make it less
  likely to be able to assign $j$ to $i$ in the future. 

\item If $(j,i)$ is a potential big-to-big move
  then the algorithm adds a big blocker $B_B$ as child to $B$ that
  consists of the machine $i$ and the big job that is assigned to $i$.
  Since $B_B$ is a big blocker this prevents us from trying to
  assign more big jobs to $i$ but at the same time allow us to try to
  assign small jobs. The intuition being
  that this will not prevent us from assigning $j$ to $i$ if the big
  job currently assigned to $i$ is reassigned.
\end{itemize}
We remark that the rules on how to update $\CompTree$ are so that a job
can be in at most one blocker whereas a machine can be in at most two
blockers (this happens if it is first added in a big blocker and then
in a small blocker).

Returning to the example in Figure~\ref{fig:idea}, we can see that
after Step~$4$, $\CompTree$ consists of the special root blocker
with two children, which in turn have a child each. Machines $M_1,
M_2$ and $M_4$ belong to small blockers whereas $M_3$ 
belongs to a big blocker. Moreover, the moves chosen in the first,
second, and the third step are big-to-small, small, and big-to-big,
respectively, and from Step $5$ to $6$ a sequence of valid moves
is chosen.

\paragraph{Values of moves.} In a specific iteration there might be several potential moves
available. For the analysis it is important that they are chosen in a
specific order. Therefore, we assign a vector in $\mathbb{R}^2$
to each move and Algorithm~\ref{algo:SimpleExtSched} will then choose
the move with minimum lexicographic value.

\begin{definition}
\label{def:simpleval}
If we let $L_i=\sigma^{-1}(i)$ then a potential move $(j,i)$ has value
$$
\val(j,i) = \left\{
\begin{array}{ll}
(0, 0) & \mbox{if valid,} \\
(1,p(L_i)) & \mbox{if small move,} \\
(2,p(L_i)) & \mbox{if  big-to-small,} \\
(3, 0) & \mbox{if big-to-big,} 
\end{array}
\right.
$$
\end{definition}

Note that as the algorithm  chooses moves of minimum lexicographic
value, it  always chooses a valid move if available and a potential
small move before a potential move of a big job.

\paragraph{The algorithm.}
Algorithm~\ref{algo:SimpleExtSched} summarizes the algorithm discussed
above in a concise definition.  Given a valid partial schedule
$\sigma$ and an unscheduled job $j_{new}$, we prove
that the algorithm preserves a valid schedule by moving
jobs until it can assign $j_{new}$.  Repeating the procedure by
choosing a unassigned job in each iteration until all jobs are
assigned then yields Theorem~\ref{thm:simple}.

\algsetup{
linenosize=\small,
linenodelimiter=:
}
\begin{algorithm*}[hbtp]
\caption{SimpleExtendSchedule($\sigma, j_{new}$)}
\label{algo:SimpleExtSched}
\vspace{0.1cm}
\begin{algorithmic}[1]
\STATE Initialize  \CompTree{} with the root  $\JO(B) = \{j_{new}\}$ and $\MA(B) = \bot$\;
\WHILE{$\sigma(j_{new})$ is  TBD} 
\STATE Choose a  potential move $(j,i)$ with $j\in \JO(\CompTree)$ of minimum lexicographic value\;
\STATE    Let $B$ be the blocker in $\CompTree$ such that $j \in \JO(B)$\;
\IF{$(j,i)$ is valid} 
\STATE Update the schedule by $\sigma(j) \leftarrow i$\;
\STATE Update $\CompTree$ by removing $B$ and all blockers added after $B$\;
\ELSIF{$(j,i)$ is either a potential small move or a potential big-to-small move}
\STATE Add  a small blocker $B_S$ as child to $B$ with $\MA(B_S) = i$ and
    $\JO(B_S) = \sigma^{-1}(i) \setminus \JO(\CompTree)$\;
\ELSE[$(j,i)$ is a big-to-big move]
\STATE Let $j_B$ be the big job such that $\sigma(j_B) = i$\;
\STATE Add a big blocker $B_B$ as a child to $B$  with $\JO(B_B) = \{j_B\}$ and  $\MA(B_B) = i$\;
\ENDIF
\ENDWHILE
\RETURN $\sigma$\;
\end{algorithmic}
\end{algorithm*}

\hide{

\begin{algorithm}[H]
  \SetKwInOut{Input}{Input}\SetKwInOut{Output}{Output}
  \SetKwFunction{Try}{TrySmallMove}
  \SetKwFunction{Update}{UpdateSchedule}
  \SetKwData{B}{Blocking}
  \BlankLine 
  Initialize  \CompTree{} with the root  $\JO(B) = \{j_{new}\}$ and $\MA(B) = \bot$\;

  \While{$\sigma(j_{new})$ is  TBD} {
    Choose a  potential move $(j,i)$ with $j\in \JO(\CompTree)$ of minimum lexicographic value\;

    Let $B$ be the blocker in $\CompTree$ such that $j \in \JO(B)$\;
    \If{$(j,i)$ is valid} {
      Update the schedule by $\sigma(j) \leftarrow i$\;
      Update $\CompTree$ by removing $B$ and all blockers added after $B$\;
  }
  \ElseIf{$(j,i)$ is either a potential small move or a potential big-to-small move}
  {
    Add  a small blocker $B_S$ as child to $B$ with $\MA(B_S) = i$ and
    $\JO(B_S) = \sigma^{-1}(i) \setminus \JO(\CompTree)$\;
  }
  \Else({$(j,i)$ is a big-to-big move}){
    Let $j_B$ be the big job such that $\sigma(j_B) = i$\;
    Add a big blocker $B_B$ as a child to $B$  with 
    $\JO(B_B) = \{j_B\}$ and  $\MA(B_B) = i$\;
  }
}
\Return $\sigma$\;
\caption{SimpleExtendSchedule($\sigma, j_{new}$):}
\label{algo:SimpleExtSched}
\end{algorithm}

}
\subsection{Analysis}
Since the algorithm only updates $\sigma$  if  a valid move was chosen we have that
the schedule stays valid throughout the execution.
It remains to verify that the algorithm terminates and that  there
always is a potential move to choose. 

Before proceeding with the proofs, we need to introduce some
notation. When arguing about  $\CompTree$ we will let
\begin{itemize}
\item $\CompTree_t$ be the subtree of $\CompTree$ induced by the   blockers $B_0, B_1, \dots, B_t$;
\item  $S(\CompTree_t) = \{j\in \JO: \mbox{$j$ is small with $\Gamma_\sigma(j)
    \subseteq \MAS(\CompTree_t)$}\}$ and often refer to $S(\CompTree)$ by simply $S$; and
\item $S_i(\CompTree_t) = \sigma^{-1}(i) \cap S(\CompTree_t)$  (often
  refer to $S_i(\CompTree)$ by  $S_i$).
\end{itemize}
The set $S(\CompTree_t)$ contains the set of small jobs with no
potential moves with respect to $\CompTree_t$. Therefore no job in
$S(\CompTree_t)$ has been reassigned since $\CompTree_t$ became a
subtree of $\CompTree$, i.e., since $B_t$ was added.  We can thus omit
the dependence on $\sigma$ when referring to $S(\CompTree_t)$ and
$S_i(\CompTree_t)$ without
ambiguity. A related
observation that will be useful throughout the analysis is the
following. No job in a blocker $B$ of $\CompTree$ has been reassigned
after $B$ was added since that would have caused the
algorithm to remove $B$ (and all blockers added after $B$).

We now continue by first proving that there always is a potential move
to choose if \CLP{} is feasible followed by the proof that the
procedure terminates in Section~\ref{sec:simpletermination}.

\subsubsection{Existence of potential moves}
\label{sec:simpleexistence}
We prove that the algorithm never gets stuck if the \CLP{} is feasible.

\begin{lemma}
\label{lemma:simplefis}
  If \CLP{} is feasible then Algorithm~\ref{algo:SimpleExtSched} can
  always choose a  potential move.
\end{lemma}
\begin{proof}
  Suppose that the algorithm has reached an iteration where no
  potential move is available.  
  We will show that this implies that the dual of \CLP{}
  is unbounded and we can thus deduce as required that the primal is
  infeasible in the case of no potential moves.

  As each solution $(y,z)$ of the dual can be scaled by a scalar
  $\alpha$ to obtain a new solution $(\alpha y, \alpha z)$, any
  solution such that $\sum_{i\in \MA} y_i < \sum_{j\in \JO} z_j$
  implies unboundedness. We proceed by defining such a solution $(y^*,
  z^*)$:
\[
z_j^* = \begin{cases}
    2/3  & \mbox{if $j\in \JO(\CompTree)$ is big,}\\
    p_j = \epsilon & \mbox{if $j \in \JO(\CompTree)\cup S$ is small,} \\
    0 & \mbox{otherwise,}
\end{cases}
\]

 and 
\[
y^*_i = \begin{cases}
  1  & \mbox{if } i \in \MAS(\CompTree),\\
  \sum_{j\in \sigma^{-1}(i)} z_j^* & \mbox{otherwise.}

\end{cases}
\]
Let us first verify that $(y^*, z^*)$ is indeed a feasible solution.
\begin{claim}
\label{claim:simplefis}
  Assuming no potential moves are available,  $(y^*, z^*)$ is a
  feasible solution.
\end{claim}
\begin{proofclaim}
  We need to verify that $y_i^* \geq \sum_{j\in C} z_j^*$ for each
  $i\in \MA$ and each $C \in \conff{i}$.  Recall that the total
  processing time of the jobs in a configuration is at most $1$. Also
  observe that $z_j^* = 0$ for jobs not in $\JO(\CompTree)\cup S$ and
  we can thus omit such jobs when verifying the constraints. 
  
  Since $z_j^* \leq p_j$ for all $j\in \JO$, we have that no
  constraint involving the variable $y_i^*$ for $i\in \MAS(\CompTree)$
  is violated. Indeed for such a machine $i$ we have $ y_i^* = 1$ and
  $\sum_{j \in \conff{i}} z_j^* \leq \sum_{j \in \conff{i}} p_j \leq
  1$ for any $C \in \conff{i}$.

  As a small job $j\in \JO(\CompTree)$ with a move $(j,i)$ is a
  potential move if $i \not \in \MAS(\CompTree)$ and no such moves exist by assumption, no small jobs in
  $\JO(\CompTree)$ can be moved to machines  in
  $\MA \setminus \MAS(\CompTree)$. Also, by definition, no small jobs in $S$ can 
  be moved to a machine $i\not\in \MAS$. This together with the fact
  that a big job $j_B$ has processing time $1$ and is thus alone in a
  configuration gives us that a constraint involving $y_i^*$ for
  $i\not \in \MAS(\CompTree)$ can only be violated if $y_i^* <
  z_{j_B}^* = 2/3$.

  As a machine $i \in \MAB(\CompTree)$
  has a big job assigned, we have that for those $y_i^* \geq 2/3$.  Now
  consider the final case when $i \not \in \MA(\CompTree)$. If a big job in $\JO(\CompTree)$
  has a move to $i$ then since it is not a potential move $p(S_i)
  > R \geq 2/3$. As $z^*_j = p_j = \epsilon$ for small jobs, we have then $y_i^* = \sum_{j\in \sigma^{-1}(i)} z_j^* \geq p(S_i) \geq 2/3$, as required.
  
  We can thus conclude that no constraint is   violated and $(y^*, z^*)$ is a feasible solution.
\end{proofclaim}

Having proved that $(y^*,z^*)$ is a feasible solution, the proof of Lemma~\ref{lemma:simplefis} is
now completed by showing that the value of the solution is negative. 
\begin{claim}
\label{claim:simplenegative}
We have that $\sum_{i\in \MA} y^*_i < \sum_{j\in \JO} z^*_j$.
\end{claim}
\begin{proofclaim}
  By the definition of $y^*$,
  \begin{equation}
    \label{eq:simpleval}
    \sum_{i\in \MA} y^*_i =
    \sum_{i\in \MAS(\CompTree)} 1 + \sum_{i \not \in \MAS(\CompTree)}
    \sum_{j\in\sigma^{-1}(i)} z_j^*
  \end{equation}
  
  We proceed by bounding $\sum_{i\in \MAS(\CompTree)} 1$ from above by
  $$\sum_{i \in \MAS(\CompTree)} \sum_{j\in \sigma^{-1}(i)} z_j^*.$$
  Let $B_0, B_1, \dots, B_\ell$ be the blockers of $\CompTree$ and
  consider a small blocker $B_t$ for some $t=1, \dots, \ell$. By the
  definition of Algorithm~\ref{algo:SimpleExtSched}, $B_t$ was added in an
  iteration when either a potential small or big-to-small move
  $(j_0,i_t)$ was chosen with $\MA(B_t) = i_t$. Suppose first that
  $(j_0,i_t)$ was a potential small move. Then as it was not valid,
  $ p(j_0) + p(\sigma^{-1}(i_t)) > 1+R.$ This inequality together
  with the fact that $i_t$ is assigned at most one big job $j_B$  gives us
  that if $(j_0, i_t)$ is a small move then
\begin{equation}
\label{eq:simplecount}
\sum_{j\in \sigma^{-1}(i_t)} z_j^*  = p(\sigma^{-1}(i_t)) - \left(p(j_B) - z_{j_B}^*\right) \geq  
1+R-p(j_0) - \frac{1}{3} = 
\frac{4}{3}.
\end{equation}

On the other hand, if $(j_0, i_t)$ is a potential
big-to-small move then as it was not valid
\begin{equation}
\label{eq:simplecount2}
\frac{2}{3}  < R <  p(\sigma^{-1}(i_t)) = \sum_{j\in \sigma^{-1}(i_t)} z_j^*,
\end{equation}
where the equality follows from that $i_t$ is only assigned small jobs
(since we assumed $(j_0, i_t)$ was a big-to-small
move). 

From~\eqref{eq:simplecount} and~\eqref{eq:simplecount2} we can see
that $\sum_{i\in \MAS(\CompTree)} 1$ is bounded from above by $\sum_{i
  \in \MAS(\CompTree)} \sum_{j\in \sigma^{-1}(i)} z_j^*$ if the number
of small blockers added because of small moves is greater than the
number of small blockers added because of big-to-small moves. 

We proceed by proving this by showing that if $B_t$ is a small blocker
added because of a potential big-to-small move then $B_{t+1}$ must be
a small blocker added because of a small move.  Indeed, the definition
of a potential big-to-small move $(j_0, i_t)$ and the fact that it was not valid
imply that
\begin{equation*}
p\left(S_{i_t}(\CompTree_t)\right) \leq R \qquad \mbox{and} \qquad p\left(\sigma^{-1}(i_t)\right) > R.
\end{equation*} 

As there are no big jobs assigned to $i_t$ (using that $(j_0, i_t)$
was a big-to-small move), the above inequalities give us that there
is always a potential small move of a small job assigned to $i_t$ with
respect to $\CompTree_{t}$.
In other words, we have that $B_t$ was not the last blocker added to
$\CompTree$ and as small potential moves have the smallest
lexicographic value (apart from valid moves), $B_{t+1}$ must be a
small blocker added because of a small move.
We can thus ``amortize'' the load of $B_{t+1}$ to increase the load of
$B_t$. Indeed, if we let $\MA(B_{t+1}) = i_{t+1}$
then~\eqref{eq:simplecount} 
and~\eqref{eq:simplecount2} yield $ \sum_{j\in \sigma^{-1}(i_t)} z_j^*
+\sum_{j\in \sigma^{-1}(i_{t+1})} z_j^* \geq 2.  $

Pairing each small blocker $B_t$ added because of a
big-to-small moves with the small blocker $B_{t+1}$ added because of a
small move as above allows us to deduce that $|\MAS(\CompTree)| \leq \sum_{i \in
  \MAS(\CompTree)} \sum_{j\in \sigma^{-1}(i)} z_j^*$. Combining this
inequality with~\eqref{eq:simpleval} yields
$$
\sum_{i\in \MA} y^*_i \leq \sum_{i\in \MA} \sum_{j\in \sigma^{-1}(i)} z_j^* = \sum_{j\in \JO} z_j^* - z_{j_{new}}^* < \sum_{j\in \JO} z_j^*,
$$
as required.
\end{proofclaim}

We have proved that there is a solution $(y^*, z^*)$ to the dual that
is feasible (Claim~\ref{claim:simplefis}) and  has negative value
(Claim~\ref{claim:simplenegative}) assuming there are no potential moves. In
other words, the \CLP{} cannot be feasible if no potential moves can be
chosen which completes the proof of the lemma.

\end{proof}

\subsubsection{Termination}
\label{sec:simpletermination}

We continue by proving that Algorithm~\ref{algo:SimpleExtSched} terminates. As
the algorithm only terminates when a new job is assigned,
Theorem~\ref{thm:simple} follows from Lemma~\ref{lemma:simpleterm} together with
Lemma~\ref{lemma:simplefis} since then we can, as already explained,
repeat the procedure until all jobs are assigned.

The intuition that the procedure terminates, assuming there always is
a potential move, is the following. As every time the algorithm
chooses a potential move that is not valid a new blocker is added
to the tree and as each machine can be in at most 2$|\MA|$ blockers,
we have that the algorithm must choose a valid move after at most
$2|\MA|$ steps. Such a move will perhaps trigger more valid moves
and each valid move makes a potential move previously blocked more
``likely''. We can now guarantee progress by measuring the
``likeliness'' in terms of the lexicographic value of the move.

\begin{lemma}
\label{lemma:simpleterm}
  Assuming there is always a potential move to choose, Algorithm~\ref{algo:SimpleExtSched} terminates.
\end{lemma}
\begin{proof}
  To prove that the procedure terminates we associate a vector, for
  each iteration, with the dynamic tree $\CompTree$. We will then show
  that the lexicographic order of these vectors decreases.
  
  The vector associated to $\CompTree$ is defined as follows. Let $B_0,
  B_1, \dots, B_\ell$ be the blockers of $\CompTree$. With
  blocker $B_i$ we will associate the value vector, denoted by
  $\val(B_i)$, of the move that was chosen in the iteration when $B_i$
  was added. The vector associated with $\CompTree$ is then simply
$$
(\val(B_0), \val(B_1), \dots, \val(B_\ell), \infty).
$$
If the algorithm adds a new blocker then  the
lexicographic order clearly decreases as the vector ends with
$\infty$.  It remains to verify what happens when blockers are
removed from $\CompTree$.  In that case let the algorithm run until it
chooses a potential move that is not valid or terminates. As
blockers will be removed in each iteration until it either
terminates or chooses a potential move that is not valid we will
eventually reach one of these cases. If the algorithm terminates we
are obviously done. 

Instead, suppose that starting with $\sigma$ and $\CompTree$ the
algorithm does a sequence of steps where blockers are removed until we
are left with an updated schedule $\sigma_k$, a tree of blockers
$\CompTree_k$ with $k+1<\ell$ blockers, and a potential move $(j',
i')$ that is not valid is chosen.  As a blocker $B$ is removed if a
blocker added earlier is removed, we have that $\CompTree_k$ equals
the subtree of $\CompTree$ induced by $B_0, B_1, \dots, B_k$.

We will thus concentrate on comparing
the lexicographic value of $(j', i')$ with that of $B_{k+1}$. Recall
that $\val(B_{k+1})$ equals the value of the move that was chosen  when
$B_{k+1}$ was added, say $(j_{t}, i_{k+1})$ for $j_{t} \in \JO(B_{t})$
with $1\leq t \leq k$ and $ \MA(B_{k+1}) = i_{k+1}$.

A key observation is that since blocker $B_{k+1}$ was removed but not
$B_k$, the most recent move was of a job $j_{k+1}\in \JO(B_{k+1})$ and we have $\sigma(j_{k+1}) = i_{k+1}$ and $\sigma_k(j_{k+1}) \neq
i_{k+1}$.  Moreover, as $(j_t, i_{k+1})$ was a potential move when
$B_{k+1}$ was added, it is a potential move with respect to
$\CompTree_k$ (using that $S_{i_{k+1}}\left(\CompTree_k\right)$ has
not changed).  Using these observations we now show that the
lexicographic value of $(j_t, i_{k+1})$ has decreased. As the
algorithm always chooses the move of minimum lexicographic value, this
will imply $\val(j',i') < \val(B_{k+1})$ as required.

If $(j_t, i_{k+1})$ was a small or big-to-small move then $B_{k+1}$
was a small blocker. As no jobs were moved to $i_{k+1}$ after
$B_{k+1}$ was added, $p(\sigma_k^{-1}(i_{k+1})) <
p(\sigma^{-1}(i_{k+1}))$ and we have that the lexicographic value
of $(j_t, i_{k+1})$ has  decreased.

Otherwise if $(j_t, i_{k+1})$ was a big-to-big move then $j_{k+1}$
must be a big job.
As we only have jobs of two sizes, the move of the big job $j_{k+1}$ implies that $(j_t, i_{k+1})$ now is a valid move which contradicts the assumption that the algorithm has chosen a potential but not valid move $(j', i')$.

We have thus proved that the vector associated to $\CompTree$ always
decreases. As there are at most $2|\MA|$ blockers in $\CompTree$
(one small and one big blocker for each machine) and a vector
associated to a blocker can take a finite set of values, we 
conclude that the algorithm terminates.
\end{proof}


\section{Proof of Main Result}
\label{sec:mainalgo}
In this section we extend the techniques presented in
Section~\ref{sec:simple} to prove our main result, i.e., that there is
a polynomial time algorithm that estimates the optimal
makespan of the restricted assignment problem within a factor of
$\frac{33}{17} + \epsilon \approx 1.9412 +\epsilon$, where $\epsilon > 0$ is
an arbitrarily small constant.  More specifically, we shall show the
following theorem which clearly implies
Theorem~\ref{thm:mainintro}. The small loss of $\epsilon$ is as
mentioned because the known polynomial time algorithms only solve \CLP{} up
to any desired accuracy.
\begin{theorem}
\label{thm:main}
The \CLP{} has integrality gap at most $\frac{33}{17}$.

\end{theorem}
Throughout this section we let $R = \frac{16}{17}$.  The proof follows
closely the proof of Theorem~\ref{thm:simple}, i.e., we design a local
search algorithm that returns a solution with makespan at most $1+R$,
assuming the \CLP{} is feasible. The strategy is again to repeatedly
call a procedure, now summarized in Algorithm~\ref{algo:ExtSched}, that extends
a given partial schedule by assigning a new job while maintaining a
valid schedule. Recall that a partial schedule is valid if each
machine $i\in \MA$ is assigned at most one big job and
$p(\sigma^{-1}(i)) \leq 1+R$, where now $R= \frac{16}{17}$. We note
that a machine is only assigned at most one big job will be a
restriction here.

Since we allow jobs of any sizes we need to define what big and small
jobs are. In addition, we  have medium jobs and partition the big
jobs into large and huge jobs.
\begin{definition}
  A job $j$ is called \emph{big} if $p_j \geq {11}/{17}$,
  \emph{medium} if $11/17 > p_j > 9/17$, and \emph{small} if $p_j \leq
  9/17$. Let the sets $\JO_B, \JO_M, \JO_S$ contain the big, medium,
  and small jobs, respectively. We will also call a big job $j$ \emph{huge}
  if $p_j \geq 14/17$ and otherwise \emph{large}.
\end{definition}
The job sizes are chosen so as to optimize the achieved
approximation ratio with respect to the analysis and were obtained by
solving a linear program.

We shall also need to extend and change some of the concepts used in
Section~\ref{sec:simple}. The final goal is still  that the
procedure shall choose potential moves so as to guarantee (i)  that the
procedure terminates and that (ii) if no potential move exists then we
shall be able to prove that the dual of \CLP{} is unbounded.  

The main difficulty compared to the case in Section~\ref{sec:simple}
of two job sizes is the following.  A key step of the analysis in the
 case with only small and big jobs was that if small jobs blocked the move of a big job then
we guaranteed that one of the small jobs on the blocking machine had a
potential small move.  This was to allow us to amortize the load when
analyzing the dual. In the general case, this might not be possible
when a medium job is blocking the move of a huge job. Therefore, we
need to introduce a new conceptual type of blockers called medium
blockers that will play a similar role as big blockers but instead of
containing a single big job they  contain at least one medium job that
blocks the move of a huge job. Rather unintuitively, we allow for
technical reasons large but not medium or huge jobs to be moved
to machines in medium blockers.

Let us now point out the modifications needed starting with the tree
$\CompTree$ of blockers.

\paragraph{Tree of blockers.} Similar to Section~\ref{sec:simple},
Algorithm~\ref{algo:ExtSched}  remembers its history by using the dynamic
tree $\CompTree$ of blockers. 
As already mentioned, it will now  also have medium
blockers that play a similar role as big blockers but instead of
containing a single big job they  contain a set of medium jobs. We
 use $\MAM(\CompTree)$ to refer to the subset of
$\MA(\CompTree)$ containing the machines in medium blockers.

As in the case of two job sizes, Algorithm~\ref{algo:ExtSched}
initializes tree \CompTree{} with the special small blocker as root
that consists of the job $j_{new}$. The next step of the procedure is
to repeatedly choose valid and potential moves until we can eventually assign
$j_{new}$.
During its execution, the procedure  now updates
$\CompTree$ based on which move that is chosen so that
\begin{enumerate}
\item $\MAS(\CompTree)$ contains those machines to which the algorithm will not try to move any jobs;
\item $\MAB(\CompTree)$ contains those machines to which the algorithm will not try to move any huge, large, or medium jobs;
\item $\MAM(\CompTree)$ contains those machines to which the algorithm will not try to move any huge or medium jobs;

\item $\JO(\CompTree)$
  contains those jobs that the algorithm wishes to move.
\end{enumerate}
We can see that small jobs are treated as in Section~\ref{sec:simple},
i.e., they can be moved to all machines apart from those in
$\MAS(\CompTree)$. Similar to big jobs in that section, huge and
medium jobs can only be moved to a machine not in $\MA(\CompTree)$.
The difference lies in how large jobs are treated: they are allowed to
be moved to machines not in $\MA(\CompTree)$ but also to those
machines only in $\MAM(\CompTree)$.

\paragraph{Potential moves.} As before a potential move $(j,i)$ will be called
\emph{valid} if  the update $\sigma(j)
\leftarrow i$ results in a valid schedule, but the definition of
potential moves needs  to be extended to include the different job
sizes. The definition for small jobs remains unchanged: a move $(j,i)$
of a small job $j \in \JO(\CompTree)$ is a potential small move if $i
\not \in \MAS(\CompTree)$.  A subset of the potential moves of  medium and
large jobs will also be called potential small moves.

\begin{definition}
  A move $(j,i)$ of a  medium or large job $j\in \JO(\CompTree)$ satisfying $i \not \in \MA(\CompTree)$ if $j$ is medium and $i \not \in \MAB(\CompTree) \cup \MAS(\CompTree)$ if $j$ is large is  a potential
  \begin{description}
  \item[\textnormal{small move:}] if $i$ is not assigned a big job.

    \item[\textnormal{medium/large-to-big:}]  if $i$ is assigned a big job.
    
\end{description}
\end{definition}
Note that the definition takes into account the convention that large
jobs are allowed to be assigned to machines not in $\MAB(\CompTree)
\cup \MAS(\CompTree)$ whereas medium jobs are only allowed to be
assigned to machines not in $\MA(\CompTree)$.  The reason for
distinguishing between whether $i$ is assigned a big (huge or large)
job will become apparent in the analysis. The idea is that if a
machine $i$ is not assigned a big job then if it blocks a move of a
medium or large job then the dual variables $z^*$ will satisfy
$\sum_{j \in\sigma^{-1}(i)} z_j^* \geq 1$, which we cannot
guarantee if $i$ is assigned a big job  since when setting
the dual variables we will round down the sizes of  big jobs which might
decrease the sum by as much as $6/17$.

It remains to define the potential moves of huge jobs.
\begin{definition}
A  move $(j, i)$ of a huge job $j\in \JO(\CompTree)$ to a machine $i \not \in \MA(\CompTree)$ is a potential
\begin{description}
\item[\textnormal{huge-to-small move}:]  if no big job is
  assigned to $i$ and $p(j) +
  p(S_i \cup M_i) \leq 1+R$;
\item [\textnormal{huge-to-big move:}] if a big job is assigned to $i$ and $p(j) +
  p(S_i \cup M_i) \leq 1+R$;

\item [\textnormal{huge-to-medium move:}] if $p(j) + p(S_i) \leq 1+R$ and
  $p(j) + p(S_i\cup M_i) > 1+R$; 
\end{description}
where
$S_i= \{j\in \sigma^{-1}(i): j \mbox{ is small with }\Gamma_\sigma(j) \subseteq \MAS(\CompTree)\}$ and $M_i = \JO_M \cap \sigma^{-1}(i)$.
\end{definition}
Again $S_i$ denotes the set of small jobs assigned to $i$ with no
potential moves with respect to the current tree. The set $M_i$
contains the medium jobs currently assigned to $i$. Moves
huge-to-small and huge-to-big correspond to the moves big-to-small and
big-to-big in Section~\ref{sec:simple}, respectively. The constraint
$p(j) + p(S_i \cup M_i) \leq 1+R$ says that such a move should only be
chosen if it can become valid by not moving any medium jobs
assigned to $i$. The additional move
called huge-to-medium covers the case when moving the medium jobs
assigned to $i$ is necessary for the move to become valid.

Similar to before, the behavior of the algorithm depends on the
type of the chosen potential move. The treatment compared to that in
Section~\ref{sec:simple} of valid moves and potential small moves is
unchanged; the huge-to-small move is treated as the big-to-small move;
and the medium/large-to-big and huge-to-big moves are both treated as
the big-to-big move were in Section~\ref{sec:simple}. It remains to
specify what Algorithm~\ref{algo:ExtSched} does in the case when a potential
huge-to-medium move (that is not valid) is chosen, say $(j,i)$ of a job $j \in \JO(B)$ for
some blocker $B$. In that case the algorithm adds a medium blocker
$B_M$ as child to $B$ that consists of the machine $i$ and the medium
jobs assigned to $i$. This prevents other huge or medium jobs to
be assigned to $i$. Also note that constraints $p(j) + p(S_i) \leq
1+R$ and $p(j) + p(S_i \cup M_i) > 1+R$ imply that there is at least
one medium job assigned to $i$.

We remark that the rules on how to update $\CompTree$ is again so that a job
can be in at most one blocker whereas a machine can now be in at most three
blockers (this can happen if it is first added in a medium blocker, then in a  big blocker, and finally in  a small blocker).
\paragraph{Values of moves.}

As in Section~\ref{sec:simple}, it is important in which order the
moves are chosen, Therefore, we assign a value vector to each
potential move and Algorithm~\ref{algo:ExtSched} chooses then, in each
iteration, the move with smallest lexicographic value.

\begin{definition}
\label{def:diffval}
If we let $L_i=\sigma^{-1}(i)$ then a potential move $(j,i)$ has value
$$
\val(j,i) = \left\{
\begin{array}{ll}
(0, 0) & \mbox{ if valid,} \\
(p(j),p(L_i)) & \mbox{ if small move,} \\
(2,0) & \mbox{ if  medium/large-to-big,} \\
(3, p(L_i)) & \mbox{ if huge-to-small,} \\
(4, 0) & \mbox{ if huge-to-big,} \\
(5, |L_i\cap \JO_M|) & \mbox{ if huge-to-medium.} 
\end{array}
\right.
$$
\end{definition}
Note that as before, the algorithm chooses a valid move if
available and a potential small move before any other potential move.
Moreover, it  chooses the potential small move of the smallest job
available.

\paragraph{The algorithm.}

Algorithm~\ref{algo:ExtSched} summarizes the algorithm concisely using
the concepts described previously. Given a valid partial schedule
$\sigma$ and an unscheduled job $j_{new}$, we shall prove that it preserves
a valid schedule by moving jobs until it can assign $j_{new}$.
Repeating the procedure until all jobs are assigned then yields
Theorem~\ref{thm:main}.

\begin{algorithm*}[hbtp]
\caption{ExtendSchedule($\sigma, j_{new}$):}
\label{algo:ExtSched}
\begin{algorithmic}[1]
\STATE Initialize \CompTree{} with the root $\JO(B) = \{j_{new}\}$ and
$\MA(B) = \bot$\;

\WHILE{$\sigma(j_{new})$ is  TBD} 
  \STATE Choose a potential move $(j,i)$ with $j\in \JO(\CompTree)$ of minimum lexicographic value\;
  \STATE Let $B$ be the blocker in $\CompTree$ such that $j \in \JO(B)$\;
  \IF{$(j,i)$ is valid} 
   \STATE Update the schedule by $\sigma(j) \leftarrow i$\;

   \STATE Update $\CompTree$ by removing $B$ and all blockers added after $B$\;
  
  \ELSIF{$(j,i)$ is either a potential small or  huge-to-small move}
  
  \STATE Add a small blocker $B_S$ as child to $B$ with  $\MA(B_S) = i$ and
    $\JO(B_S) = \sigma^{-1}(i)\setminus \JO(\CompTree)$;
  \ELSIF{$(j,i)$ is either a potential large/medium-to-big or huge-to-big move} 
  \STATE Let $j_B$ be the big job such that $\sigma(j_B) = i$\;
   \STATE Add a big blocker $B_B$ as a child to $B$ in \CompTree{} with 
    $\JO(B_B) = \{j_B\}$ and  $\MA(B_B) = i$\;

  \ELSE[$(j,i)$ \emph{is a potential huge-to-medium move}] 
    \STATE Add a medium blocker $B_M$ as child to $B$  with  $\MA(B_M) = i$ and $\JO(B_M) =   \sigma^{-1}(i) \cap \JO_M$\;
    \ENDIF
\ENDWHILE
\RETURN $\sigma$\;

\end{algorithmic}
\end{algorithm*}

\hide{
\begin{procedure}[H]
  \SetKwInOut{Input}{Input}\SetKwInOut{Output}{Output}
  \SetKwData{Tree}{ComponentTree}
  \SetKwData{SM}{PossibleMoves}
  \SetAlgoNoEnd
  \SetKwFunction{Try}{TrySmallMove}
  \SetKwFunction{Update}{UpdateSchedule}
  \SetKwData{B}{Blocking}
Initialize  \CompTree{} with the root  $\JO(B) = \{j_{new}\}$ and $\MA(B) = \bot$\;

\While{$\sigma(j_{new})$ is  TBD} {
  Choose a potential move $(j,i)$ with $j\in \JO(\CompTree)$ of minimum lexicographic value\;

  Let $B$ be the blocker in $\CompTree$ such that $j \in \JO(B)$\;
  \If{$(j,i)$ is valid} {
    Update the schedule by $\sigma(j) \leftarrow i$\;

    Update $\CompTree$ by removing $B$ and all blockers added after $B$\;
  }
  \ElseIf{$(j,i)$ is either a potential small or  huge-to-small move}
  {
    Add a small blocker $B_S$ as child to $B$ with  $\MA(B_S) = i$ and
    $\JO(B_S) = \sigma^{-1}(i)\setminus \JO(\CompTree)$;
  }
  \ElseIf{$(j,i)$ is either a potential large/medium-to-big or huge-to-big move} {
    Let $j_B$ be the big job such that $\sigma(j_B) = i$\;
    Add a big blocker $B_B$ as a child to $B$ in \CompTree{} with 
    $\JO(B_B) = \{j_B\}$ and  $\MA(B_B) = i$\;
  }
  \Else({$(j,i)$ \emph{is a potential huge-to-medium move}}) {Add a medium blocker $B_M$ as child to $B$  with  $\MA(B_M) = i,$ and $\JO(B_M) =   \sigma^{-1}(i) \cap \JO_M$\;} 
}
\Return $\sigma$\;
\caption{ExtendSchedule($\sigma, j_{new}$):}
\label{algo:ExtSched}
\end{procedure}
}

\subsection{Analysis}

As Algorithm~\ref{algo:SimpleExtSched} for the case of two job sizes,
Algorithm~\ref{algo:ExtSched} only updates the schedule if a valid
move is chosen so it follows that the schedule stays valid throughout
the execution. It remains to verify that the algorithm terminates and
that the there always is a potential move to choose.

The analysis is similar to that of the simpler case but involves more
case distinctions.
 When arguing
about $\CompTree$ we again let

\begin{itemize}
\item $\CompTree_t$ be the subtree of $\CompTree$ induced by the   blockers $B_0, B_1, \dots, B_t$;
\item  $S(\CompTree_t) = \{j\in \JO: \mbox{$j$ is small with $\Gamma_\sigma(j)
    \subseteq \MAS(\CompTree_t)$}\}$ and often refer to $S(\CompTree)$ by simply $S$; and
\item $S_i(\CompTree_t) = \sigma^{-1}(i) \cap S(\CompTree_t)$  (often
  refer to $S_i(\CompTree)$ by  $S_i$).
\end{itemize}
As in the case of two job sizes, the set $S(\CompTree_t)$ contains the set
of small jobs with no potential moves with respect to
$\CompTree_t$. Therefore no job in $S(\CompTree_t)$ has been
reassigned since $\CompTree_t$ became a subtree of $\CompTree$, i.e.,
since $B_t$ was added.  We can thus again omit without ambiguity the dependence on $\sigma$
when referring to $S(\CompTree_t)$ and $S_i(\CompTree_t)$. 
The following observation will again be used  throughout the analysis. No job in a blocker $B$ of $\CompTree$ has been reassigned
after $B$ was added since that would have caused the
algorithm to remove $B$ (and all blockers added after $B$).

%
\hide{
Observation~\ref{obs:simple} we have the following from the
fact that Algorithm~\ref{algo:ExtSched} removes a blocker either (i) if any
blocker added before it is removed from \CompTree{} or (ii) if a job in the
blocker is reassigned.
\begin{observation}
\label{obs:diff}
Procedure~\ref{algo:ExtSched} satisfies the following
invariant. Order the blockers $B_0,B_1, B_2, \dots, B_\ell$ of
$\CompTree$ in the order they were added ($B_0$ is the special root
blocker). If we let $\CompTree_t$ denote the subtree of $\CompTree$
induced by $B_0, B_1,\dots, B_t$ then for $t=1,\dots, \ell$
\begin{enumerate} 
\item in the iteration when $B_t$ was added,  the tree of blockers equaled $\CompTree_{t-1}$ at the beginning and $\CompTree_t$ at the end of that iteration; 
\item each job $j\in \JO(\CompTree_{t})$ has not been reassigned  after $B_{t}$ was added to the tree of blockers.
\end{enumerate}
\end{observation}
}
%
 In Section~\ref{sec:existence} we start by
 presenting
  the proof that the algorithm always can choose a potential move if
  \CLP{} is feasible. We then present the proof that the algorithm
  always terminates in Section~\ref{sec:termination}.  As already
  noted, by repeatedly calling Algorithm~\ref{algo:ExtSched} until all
  jobs are assigned we will, assuming \CLP{} feasible, obtain a
  schedule of the jobs with makespan at most $1+R$ which completes the
  proof of Theorem~\ref{thm:main}.

\subsubsection{Existence of potential moves}
\label{sec:existence}
We prove that the algorithm never gets stuck if the \CLP{} is feasible.
\begin{lemma}
  \label{lemma:existence}
  If \CLP{} is feasible then Algorithm~\ref{algo:ExtSched} can
  always pick a  potential move.
\end{lemma}
\begin{proof}
  Suppose  that the algorithm has reached an iteration where no  potential move is available.
  Similar to the proof of Lemma~\ref{lemma:simplefis} we will show
  that this implies that the dual is unbounded and hence the primal is
  not feasible. We will do so by defining a solution $(y^*, z^*)$ to
  the dual with $\sum_{i\in \MA} y^*_i < \sum_{j\in \JO}
  z^*_j$. Define $z^*$ and $y^*$ by
  $$
  z^*_j = \begin{cases}
    11/17,  & \mbox{if $j\in \JO(\CompTree)$ is big},\\
    9/17, & \mbox{if $j \in \JO(\CompTree)$ is medium}, \\
    p_j, & \mbox{if $j \in \JO(\CompTree)\cup S$ is small}, \\
    0, & \mbox{otherwise,}
  \end{cases}
$$
and
$$
  y^*_i = \begin{cases}
    1 & \text{if $i \in \MAS(\CompTree)$}, \\
    \sum_{j\in \sigma^{-1}(i)} z_j^* & \text{otherwise.}
  \end{cases}
  $$
  A natural interpretation of $z^*$ is that we have rounded down processing
times where big jobs with processing times in $[11/17,1]$ are rounded
down to $11/17$; medium jobs with processing times in $[9/17, 11/17]$
are rounded down to $9/17$; and small jobs are left unchanged.

The proof of the lemma is now completed by showing that $(y^*, z^*)$ is a feasible solution (Claim~\ref{claim:compfis}) and that the objective value is negative (Claim~\ref{claim:compneg}).
  \begin{claim}
    \label{claim:compfis}
    Assuming no potential moves are available, $(y^*, z^*)$ is a feasible solution.
  \end{claim}
  \begin{proofclaim}
    We need to verify that $y_i^* \geq \sum_{j\in C} z_j^*$ for each
  $i\in \MA$ and each $C\in \conff{i}$. Recall that the total
  processing time of the jobs in a configuration is at most $1$.  Also
  observe as in the simpler case that $z_j^*=0$ for jobs not in
  $\JO(\CompTree) \cup S$ and we can thus omit them when verifying the
  constraints.

  As in the previous case, no constraint involving the variable
  $y_i^*$ for $i\in \MAS(\CompTree)$ is violated. Indeed, for such a
  machine $i$ we have $y_i^* = 1$ and  $\sum_{j\in \conff{i}} z_j^* \leq \sum_{j\in \conff{i}} p_j
  \leq 1$ for any configuration $C \in
  \conff{i}$.

  We continue by distinguishing between the remaining cases when
  $i\not \in \MA(\CompTree), i\in \MAB(\CompTree)\setminus
  \MAS(\CompTree)$, and $i\in \MAM(\CompTree)\setminus
  (\MAS(\CompTree) \cup \MAB(\CompTree))$. 
  \paragraph{$i \not \in \MA(\CompTree)$:} A move $(j,i)$ to machine
  $i\not\in \MA(\CompTree)$ is a potential move if $j\in
  \JO(\CompTree)$ is a small, medium, or large job. By assumption we
  have thus that no such jobs in $\JO(\CompTree)$ have any moves to
  $i$. We also have, by definition, that no job in $S$ has a move to
  $i\not \in \MAS(\CompTree)$. 

  Now consider the final case when a huge job $j_H \in \JO(\CompTree)$ has a
  move $(j_H,i)$. Then since it is not a potential huge-to-small, huge-to-medium, or huge-to-big move
$$p(j_H) + p(S_i) > 1+R \qquad \mbox{and hence} \qquad p(S_i) > R.$$
 Since $p(j_H) \geq {14}/{17}$ but $z_{j_H}^* = 11/17$ a configuration $C\in \conff{i}$ with $j_H\in C$
satisfy
\begin{equation}
\label{eq:mastertwo}
\sum_{j\in C} z_j^* \leq z_{j_H}^* + (1-p(j_H)) \leq \frac{11}{17} + \left(1- \frac{14}{17}\right) = \frac{14}{17},
\end{equation}
which is less than $ \frac{16}{17} = R < p(S_i) \leq \sum_{j\in \sigma^{-1}(i)} z_j^* \leq y^*_i$ where we used that $S_i$ only contains small
jobs for the second inequality.

\paragraph{$i \in \MAB(\CompTree)\setminus \MAS(\CompTree)$:} A move
to machine $i\not \in \MAS(\CompTree)$ is a potential move if $j\in
\JO(\CompTree)$ is a small job. By the assumption that there are no
potential moves and by the definition of $S$, we have thus that there
is no small job in $\JO(\CompTree) \cup S$ with a move to $i$. In other words, all small jobs with
  $z_j^* >0$ that can be moved to $i$ are already assigned to $i$.
Now let
$j_B$ be the big job such that $\sigma(j_B) = i$ which must exist
since machine $i$ is contained in a big blocker. As both medium and big (large and
huge) jobs have processing time strictly greater than $1/2$, a
configuration $C\in \conff{i}$ can contain at most one such job
$j_0$. Since $z_{j_0}^* \leq 11/17$ and $z_{j_B}^* = 11/17$ we have
that  a configuration $C\in \conff{i}$ with $j_0\in C$ cannot violate feasibility. Indeed,
  \begin{equation}
    \label{eq:master}
  \sum_{j\in C} z_j^* \leq z_{j_0}^* + \sum_{j\in
    \sigma^{-1}(i)\setminus j_B} z_j^* \leq z_{j_B}^* + \sum_{j\in
    \sigma^{-1}(i)\setminus j_B} z_j^* = y^*_i.
  \end{equation}

\paragraph{$i\in \MAM(\CompTree)\setminus (\MAS(\CompTree)\cup
  \MAB(\CompTree)):$} By the definition of $S$ and  the assumption that there are no potential
  moves, no small or large jobs in $\JO(\CompTree)\cup S$ have any moves to
  $i$.   As $i$ is contained in a medium blocker, there is a medium job $j_M$ such that $\sigma(j_M) = i$.
  As both medium and huge jobs have processing time strictly greater
  than $1/2$, a configuration $C\in \conff{i}$ can contain at
  most one such job $j_0$. 

  If $j_0$ is medium then $z_{j_0}^* \leq z_{j_M}^*$ and we have by
  Inequality~\eqref{eq:master} (substituting $j_B$ with $j_M$) that
  the constraint is not violated.  Otherwise if $j_0$ is huge then
  from~\eqref{eq:mastertwo}
  (substituting $j_H$ with $j_0$) we have that $\sum_{j\in C} z_j^* \leq
  \frac{14}{17}$ for any configuration $C\in \conff{i}$ with $j_0 \in C$.

  We shall now show that $y_i^* \geq \frac{14}{17}$ which completes
  this case. If $i$ is assigned more than one medium job then  $y_i^* \geq 1$.
  Instead suppose that the only medium job assigned to $i$ is
  $j_M$. Since the medium blocker $B_M$, with $\MA(B_M) =
  i$ and  $\JO(B_M) = \{j_M\}$, was added to $\CompTree$ because the algorithm chose a
  potential  huge-to-medium move we have
$$
p(S_i(\CompTree')) + p(j_M) > R.$$
where $\CompTree'$ refer to the tree of blockers  at the time when $B_M$ was
added.

As each blocker in $\CompTree'$ is also in $\CompTree$, $p(S_i) \geq p(S_i(\CompTree'))$ and  the above inequality
yields,
$$
y_i^* \geq p(S_i) + p(j_M) - (p(j_M) - z_{j_M}^*) \geq  R - \left(\frac{11}{17} - \frac{9}{17}\right) \geq \frac{14}{17},
$$
as required.

We have thus verified all the cases which completes the proof of the
claim.
\end{proofclaim}

Having verified that $(y^*, z^*)$ is a feasible solution, the proof is now completed by showing that the value of the solution is negative.

\begin{claim}
\label{claim:compneg}
We have that  $\sum_{i\in \MA} y^*_i < \sum_{j \in \JO} z_j^*$.
\end{claim}
\begin{proofclaim}
 By the definition of $y^*$, 
\begin{equation}
\label{eq:diffval}
\sum_{i\in \MA} y^*_i =
\sum_{i\in \MAS(\CompTree)} 1 + \sum_{i \not \in \MAS(\CompTree)}
\sum_{j\in\sigma^{-1}(i)} z_j^*.
\end{equation}
Similar to the  case with two job sizes in Section~\ref{sec:simple}, we proceed by
bounding $\sum_{i\in \MAS(\CompTree)} 1$ from above by $\sum_{i \in
  \MAS(\CompTree)} \sum_{j\in \sigma^{-1}(i)} z_j^*$. 
  Let $B_0, B_1, \dots, B_\ell$ be the blockers of $\CompTree$ and
  consider a small blocker $B_t$ for some $t=1, \dots, \ell$ with $\MA(B_t) = i_t$.

By the definition of the procedure, small
blocker $B_t$ has been added in an iteration when either a potential
small or huge-to-small move $(j_0,i_t)$ was chosen.  Let us
distinguish between the three cases when $(j_0, i_t)$ was a small move
of a small job, $(j_0, i_t)$ was a small move of a medium or large job,
and $(j_0, i_t)$ was a huge-to-small move.

\paragraph{$(j_0, i_t)$ was a small move  of  a small job:} Since the move was not valid,
$$
p(j_0) + p(\sigma^{-1}(i_t)) > 1+R.
$$
As a big job with processing time in $[11/17, 1]$ is rounded down to
$11/17$ and a medium job with processing time in $[9/17, 11/17]$ is
rounded down to $9/17$, the sum
$\sum_{j\in \sigma^{-1}(i_t)} z_j^*$ will depend of the number of big
and medium jobs assigned to $i_t$. Indeed, if we let $L_{i_t} =
\sigma^{-1}(i_t) \cap \JO_B$ and $M_{i_t} = \sigma^{-1}(i_t) \cap
\JO_M$ be the big and medium jobs assigned to $i_t$,
respectively. Then on the one hand,
$$
\sum_{j\in \sigma^{-1}(i_t)} z_j^* \geq 1+R - p(j_0) -
\left(1-\frac{11}{17}\right)|L_{i_t}| - \left(\frac{11}{17} - \frac{9}{17}\right)|M_{i_t}|,
$$
which equals
$$ \frac{33}{17} - \frac{6}{17}\left(|L_{i_t}| + \frac{|M_{i_t}|}{3}\right) - p(j_0).
$$
On the other hand, if (i) $|L_{i_t}| \geq 2$,  (ii) $|L_{i_t}| \geq 1$ and $|M_i| \geq 1$ or (iii)
$|M_{i_t}| \geq 3$ then clearly 
$$\sum_{j\in \sigma^{-1}(i_t)} z_j^* \geq
\frac{11}{17} + \frac{9}{17} \geq \frac{20}{17}.
$$
Combining these two bounds we get (since $j_0$ is small and thus $p(j_0) \leq 9/17$) that
\begin{align}
\label{eq:jsmall}
\sum_{j\in \sigma^{-1}(i_t)} z_j^* &\geq  \min \left[\frac{27}{17}- p(j_0),  \frac{20}{17}\right] \geq \frac{18}{17} & \mbox{if $(j_0,i_t)$ was a  small move of a small job.}
\end{align}

\paragraph{$(j_o, i_t)$ was a small move of a medium or large job:}Since the move was not valid, we have again 
$$
p(j_0) + p(\sigma^{-1}(i_t)) > 1+R.
$$
By the definition of potential small moves it must be that there is no big job assigned to $i_t$. If there are more than one medium job assigned to $i$  then clearly $\sum_{j\in \sigma^{-1}(i_t)} z_j^* \geq 1$. 

Otherwise, if there is  at most one medium job assigned to $i_t$ that might have been rounded down from $11/17$ to $9/17$ we use the inequality $p(j_0) + p(\sigma^{-1}(i_t)) > 1+R$ to derive
$$
\sum_{j\in \sigma^{-1}(i_t)} z_j^*  \geq 1+ R -p(j_0) - \left(\frac{11}{17}- \frac{9}{17}\right) = \frac{31}{17} - p(j_0)\geq 1,
$$
For the final inequality we used that the processing time of medium
and large jobs is at most ${14}/{17}$.  Summarizing again, we have
\begin{align}
\sum_{j\in \sigma^{-1}(i_t)} z_j^* &\geq  1 & \mbox{if $(j_0,i_t)$ was a potential small move of a medium or large job.}
\end{align}

\paragraph{$(j_0, i_t)$ was a huge-to-small move:}
Similar to the  case of two job sizes in Section~\ref{sec:simple}  where we considered the big-to-small move, we will  show that we can amortize the cost from the blocker $B_{t+1}$. Indeed, since the move $(j_0, i_t)$ was not valid and by the definition of a potential huge-to-small move we have
\begin{equation}
\label{eq:diffar}
p(j_0) + p(\sigma^{-1}(i_t)) > 1+ R \qquad \mbox{and} \qquad
p(j_0) + p\left(S_{i_t}(\CompTree_t) \cup M_{i_t}\right) \leq 1+R,
\end{equation}
where we let $M_{i_t}$ contain the medium jobs
assigned to $i_t$ at the time when $B_t$ was added.
Recall that $S_{i_t}(\CompTree_t)$ contains those small jobs that have
no potential moves with respect to the tree of blockers $\CompTree_t$
after $B_t$ was added. As $B_t$ has not
been removed from $\CompTree$ both these sets have not changed. In
particular, since $i_t$ is not assigned a big job (using the
definition of huge-to-small move) we have that there must be a small
job $j'\in \sigma^{-1}(i_t) \setminus (S_{i_t}(\CompTree_t) \cup M_{i_t})$ that has a potential small $(j', i')$ move with respect to
$\CompTree_t$.

The above discussion implies that $t< \ell$. Moreover, the value of
$(j',i')$ equals $(p(j'), p(\sigma^{-1}(i')))$. As the procedure
always chooses a potential move with minimum lexicographic value we have that
$B_{t+1}$ is a small blocker added because a potential small move was
chosen of a small job with processing time at most $p(j')$. 

If we let $\MA(B_{t+1}) = i_{t+1}$ we have thus
by~\eqref{eq:jsmall}
\begin{equation}
\label{eq:nyy}
\sum_{j\in \sigma^{-1}(i_{t+1})}z_j^* \geq  \min \left[\frac{27}{17}- p(j'),  \frac{20}{17}\right] \geq \frac{18}{17}.
\end{equation}
We proceed by showing that we can use this fact to again amortize the
cost as done in the simpler analysis, i.e., that
\begin{equation}
\label{eq:ammortize}
\sum_{j\in \sigma^{-1}(i_{t})}z_j^*
+ \sum_{j\in \sigma^{-1}(i_{t+1})} z_j^* \geq 2.
\end{equation}
For this reason, let us distinguish between three subcases depending on the number of medium jobs assigned to $i_t$.
\begin{description}
\item[\textnormal{$i_t$ is assigned at least two medium jobs:}] Since
medium jobs have value $9/17$ in the dual this clearly implies that
$\sum_{j\in \sigma^{-1}(i_t)} z_j^* \geq 18/17$ so no amortizing is needed  and~\eqref{eq:ammortize} holds in this case.

\item[\textnormal{$i_t$ is assigned one medium job:}] Let $j_M$ denote
  the medium job assigned to $i_t$.  In addition we have that the
  small job $j'$ is assigned to $i_t$. We have thus that $\sum_{j\in
    \sigma^{-1}(i)} z_j^* \geq 9/17 + p(j')$. So if $p(j') \geq 7/17$ then
  we can amortize because $\sum_{j\in \sigma^{-1}(i_{t+1})} z_j^*$ is at least
  $18/17$.

Now consider the case when $p(j') < 7/17$. By the assumption that only one
medium job is assigned to $i_t$ and since $p(j_0) +
p(\sigma^{-1}(i_t)) \geq 1+ R$ we have
$$
\sum_{j\in \sigma^{-1}(i_t)} z_j^* \geq R - \left(\frac{11}{17} -
  \frac{9}{17}\right) = \frac{14}{17}.
$$
Moreover, as  the case when $p(j') < 7/17$ is considered, we have from~\eqref{eq:nyy} that
$$
\sum_{j\in \sigma^{-1}(i_{t+1})} z_j^* \geq \frac{20}{17},
$$which implies that~\eqref{eq:ammortize} is valid also in this case.

\item[\textnormal{$i_t$ is not assigned any medium jobs:}] As $i_t$
  in this case is not assigned any medium or big jobs and  the
  huge-to-small move $(j_0,i_t)$ was not valid,
$$
 \sum_{j\in \sigma^{-1}(i_t)} z_j^*  > R  \geq \frac{16}{17}.
$$
That we can amortize now follows from~\eqref{eq:nyy} that
we always have
$
\sum_{j\in \sigma^{-1}(i_{t+1})} z_j^* \geq \frac{18}{17}.
$ 
\end{description}

We have thus shown that a blocker $B_t$ added because of a
huge-to-small move can be paired with the following blocker $B_{t+1}$
that was added because of a potential small move of a small job. As seen above
this give us that we can amortize the cost exactly as done in the
simpler case and
we deduce that $|\MAS(\CompTree)| \leq \sum_{i \in
  \MAS(\CompTree)} \sum_{j\in \sigma^{-1}(i)} z_j^*$. Combining this
inequality with~\eqref{eq:simpleval} yields
$$
\sum_{i\in \MA} y^*_i \leq \sum_{i\in \MA} \sum_{j\in \sigma^{-1}(i)} z_j^* = \sum_{j\in \JO} z_j^* - z_{j_{new}}^* < \sum_{j\in \JO} z_j^*,
$$
as required.
\end{proofclaim}

We have proved that there is a solution $(y^*, z^*)$ to the dual that
is feasible (Claim~\ref{claim:compfis}) and  has negative value
(Claim~\ref{claim:compneg}) assuming there are no potential moves. In
other words, the \CLP{} cannot be feasible if no potential moves can be
chosen which completes the proof of the lemma.
\end{proof}

\subsubsection{Termination}
\label{sec:termination}

As the algorithm only terminates when a new job is assigned,
Theorem~\ref{thm:main} follows from the lemma below together with
Lemma~\ref{lemma:existence} since then we can, as already explained,
repeat the procedure until all jobs are assigned.


\begin{lemma}
  Assuming there is always a potential move to choose,
  Algorithm~\ref{algo:ExtSched} terminates.
\end{lemma}
\begin{proof}
  As in the proof of Lemma~\ref{lemma:simpleterm}, we consider the $i$'th iteration of
  the algorithm and associate the vector
$$
(\val(B_1), \val(B_2), \dots, \val(B_\ell), \infty).
$$
with $\CompTree$ where the blockers are ordered in the order they were added and $\val(B_i)$ equals the value of the move chosen when $B_i$ was added.
We shall now prove that the value of the vector associated with $\CompTree$ decreases no matter the step chosen in the next iteration.

If a new blocker is added in the $i+1$'th iteration the
lexicographic order clearly decreases as the vector ends with
$\infty$.  It remains to verify what happens when blockers are
removed from $\CompTree$.  In that case let the algorithm run until it
chooses a potential move that is not valid or terminates. As
blockers will be removed in each iteration until it either
terminates or chooses a potential move that is not valid we will
eventually reach one of these cases. If the algorithm terminates we
are obviously done.  

Instead, suppose that starting with $\sigma$ and $\CompTree$ the
algorithm does a sequence of steps where blockers are removed until
we are left with an updated schedule $\sigma_k$, a $\CompTree_k$ with
$k+1< \ell$ blockers, and a potential
move $(j', i')$ that is not valid is chosen.  
As a blocker $B$ is removed if a
blocker added earlier is removed, we have that $\CompTree_k$ equals
the subtree of $\CompTree$ induced by $B_0, B_1, \dots, B_k$.

We will
thus concentrate on comparing the lexicographic value of $(j', i')$
with that of $B_{k+1}$. The value of $B_{k+1}$ equals of the value of
the move chosen when $B_{k+1}$ was added, say $(j_{t}, i_{k+1})$ for
 $j_{t} \in  \JO(B_t)$  with $1 \leq t \leq k$ and $\MA(B_{k+1}) = i_{k+1}$.

A key observation is that since blocker $B_{k+1}$ was removed but
not $B_k$, the most recent move was of a job $j_{k+1}\in \JO(B_{k+1})$ and we have $\sigma(j_{k+1}) = i_{k+1}$ and $\sigma_k(j_{k+1}) \neq i_{k+1}$.
Moreover, as $(j_t, i_{k+1})$ was a potential move when
$B_{k+1}$ was added, it is a potential move with respect to
$\CompTree_k$ (using that $S_{i_{k+1}}\left(\CompTree_k\right)$ has
not changed).  Using these observations we now show that the
lexicographic value of $(j_t, i_{k+1})$ has decreased. As the
algorithm always picks the move of minimum lexicographic value, this
will imply $\val(j',i') < \val(B_{k+1})$ as required.

If $(j_t, i_{k+1})$ was a small or huge-to-small move then $B_{k+1}$
was a small blocker. As no jobs were moved to $i_{k+1}$ after
$B_{k+1}$ was added, $p(\sigma_k^{-1}(i_{k+1})) <
p(\sigma^{-1}(i_{k+1})$ and we have that the lexicographic value
of $(j_t, i_{k+1})$ has  decreased.

If $(j_t, i_{k+1})$ was a medium/large-to-big (or huge-to-big move)  then $j_{k+1}$ must
be a big job and hence the machine $i_{k+1}$ has no longer a big job
assigned. This implies that the move $(j_t, i_{k+1})$ is now a
potential small (or huge-to-small using that no medium jobs are moved to machines in big blockers) with smaller lexicographic value.

Finally, if $(j_t, i_{k+1})$ was a huge-to-medium move then $j_{k+1}$
must be a medium job and since no medium jobs have potential moves to
machines in medium blockers we have that $|\sigma_k^{-1}(i) \cap
\JO_M| < |\sigma^{-1}(i) \cap \JO_M|$ which implies that the
lexicographic value of $(j_t, i_{k+1})$ also decreased in this case.

We have thus proved that the vector associated to $\CompTree$ always
decreases. As there are at most $3|\MA|$ blockers in $\CompTree$
(one small, medium, and big blocker   for each machine) and a vector
associated to a blocker can take a finite set of values, we 
conclude that the algorithm eventually terminates.
\end{proof}

\section{Conclusions}
We have shown that the configuration LP gives a polynomial time
computable lower bound on the optimal makespan that is strictly better
than two. Our techniques are mainly inspired by recent developments on
the related Santa Claus problem and gives a local search algorithm to
also find a schedule of the same performance guarantee, but is not
known to converge in polynomial time.

Similar to the Santa Claus problem, this raises the open question
whether there is an efficient rounding of the configuration LP that
matches the bound on the integrality gap (see also~\cite{FFeige08} for
a comprehensive discussion on open problems related to the difference
between estimation and approximation algorithms). Another interesting
direction is to improve the upper or lower bound on the integrality
gap for the restricted assignment problem: we show that it is no worse
than $33/17$ and it is only known to be no better than $1.5$ which
follows from the NP-hardness result. One possibility would be to find
a more elegant generalization of the techniques, presented in
Section~\ref{sec:simple} for two job sizes, to arbitrary processing times
(instead of the exhaustive case distinction presented in this paper).

To obtain a tight analysis, it would be natural to start with  the special
case of graph balancing for which the $1.75$-approximation algorithm
by Ebenlendr et al.~\cite{EKS08} remains the best known.  We remark that the
restriction $p_{ij} \in \{p_j, \infty\}$ is necessary as the
integrality gap of the configuration LP for the general case is known to be $2$ even if a job can be assigned to at most $2$ machines~\cite{VW10}.

\section{Acknowledgements}
I am grateful to Johan H\aa stad for many useful insights and
comments. I also wish to thank Tobias M\"{o}mke and Luk\'{a}\v{s}
Pol\'{a}\v{c}ek for useful comments on the exposition.  This research
is supported by ERC Advanced investigator grant 226203.

\bibliographystyle{abbrv}
\bibliography{main}  
\end{document}

%% file: envs.tex
\newtheorem{theorem}{Theorem}[section]

\newtheorem{observation}[theorem]{Observation}
\newtheorem{lemma}[theorem]{Lemma}
\newtheorem{claim}[theorem]{Claim}

\newtheorem{definition}[theorem]{Definition}

\newenvironment{proofclaim}{\begin{trivlist}
    \item[\hskip\labelsep {\it Proof of Claim}.]}{\QED \end{trivlist}}
\newenvironment{proof}{\begin{trivlist}
    \item[\hskip\labelsep {\bf Proof}.]}{\QED \end{trivlist}}

\newcommand{\QED}{\hfill $\square$}

\newcommand{\hide}[1]{
}

\newcounter{myclaim}